\documentclass{llncs}

\title{A Variant of Earley Deduction With \mbox{Partial Evaluation}}
\titlerunning{A Variant of Earley Deduction}
\author{Heike Stephan \and Stefan Brass} 
\authorrunning{H. Stephan \and S. Brass}
\institute{Martin-Luther-Universit\"at Halle-Wittenberg,
	Institut f\"ur Informatik,\\
	Von-Seckendorff-Platz~1, D-06099 Halle (Saale), Germany\\
	\email{stephan@informatik.uni-halle.de}, \email{brass@informatik.uni-halle.de}}

\input{defWLP12}

\begin{document}
\maketitle
\begin{abstract}
We present an algorithm
for query evaluation
given a logic program consisting of function-free Datalog rules. 
It is based on Earley Deduction~\cite{PW83,Por86}
and uses a partial evaluation
similar to the one we developed for our SLDMagic
method~\cite{Bra00}. With this, finite automata modeling the evaluation of given queries are generated.
In certain cases,
the new method is more efficient than SLDMagic
and the standard Magic Set method since it can process several deduction steps as one.
\end{abstract}

\section{Introduction}
\label{secIntro}
The goal of deductive database systems is
to offer integrated systems
which permit to do programming and database tasks
in a single, declarative language.
This would improve the current situation
in which several languages are mixed, e.\,g.\@ Java and SQL.
While SQL is declarative
and has successfully shown the advantages of declarative languages,
the language used for application programming is usually non-declarative.

Whereas in earlier times
deductive database research was concentrated
only on recursive query evaluation,
now new applications,
e.\,g.~for the semantic web,
are in the focus.
But even so mundane tasks as the generation of web pages
must be considered if deductive databases
should be used for real world application programming.
In~\cite{Bra12} we made a proposal
for declarative output
and also investigated
how sorting can be integrated in Datalog, which is obviously important for output
and also for database queries.

However,
even the classical task of bottom-up query evaluation
deserves more research in order to improve efficiency~\cite{Bra10CPP}. 
The new method presented in this paper  
loops through a sequence of states
(sets of rules being processed). From one state to the next, the successor state is determined by a single database fact and the preceding rule set. At compilation time, when the database state is not yet known, partial evaluation of the program is done by using facts with abstract values. The method is based on Earley Deduction~\cite{PW83,Por86} which exploits the similarity of context free grammar rules and rules of logic programs.
The partial evaluation
is similar to the one we developed for our SLDMagic
method~\cite{Bra00}.
It makes the algorithm very competetive
because program analysis and abstract execution can already be done at compilation time.

The algorithm is also interesting
because it especially fits applications
in which input must be parsed
(after all, the Earley algorithm is a parsing algorithm).
An input text, e.\,g.~{\tt"abc"}, can be represented by Datalog facts
as follows:
\begin{quote}
\begin{tt}
\begin{tabular}{@{}l@{}}
input(1, a, 2).\\
input(2, b, 3).\\
input(3, c, 4).\\
eof(4).
\end{tabular}
\end{tt}
\end{quote}
This is very similar to the standard difference list technique
for definite clause grammars,
but since Datalog has no lists we use positions in the input.

First we give some basic definitions (Section \ref{sec:basic}), then we describe the basic deducion method (Section \ref{sec:deduct}) and finally present the partial evaluation (Section \ref{sec:partEval}).
\section{Basic Definitions}
\label{sec:basic}
\begin{definition}[Rule]
A rule is a formula of the form $\litA \leftarrow \litB_{1} \andUntil \litB_{n}$ where $\litA$ and $\litB_{i}, i = 1, \ldots, n,$ are positive literals, i.\,e. atomic formulas $\id p(t_{1}, \ldots, t_{m})$ with a predicate $\id p$ of arity $m$ and terms $t_{j}, j = 1, \ldots, m$. Terms are variables or constants. In the above rule, $\litA$ is called the head and $\litB_{1} \andUntil \litB_{n}$ is called the body. A rule with empty body (i.\,e. $n = 0$) and without variables is called a fact.
\end{definition}

In the context of deductive databases, the \emph{range restriction} condition ensures that no derived fact contains variables. In the remaining paper, this condition is assumed to be satisfied for every rule.

\begin{definition}[Range Restriction]
A rule is range restricted iff every variable that appears in the head appears also in the body.
\end{definition}

\begin{definition}[EDB- and IDB-Predicates, Program, and Database]
Predicates are partitioned into EDB (``extensional database'') predicates defined by facts and IDB (``intensional database'') predicates defined by rules. A logic program is a finite set of rules with an IDB predicate in the head and at least one body literal. A database is a finite set of facts with EDB predicate.
\end{definition}

The requirement that the body of a program rule is non-empty simplifies later definitions but is no restriction: One can use a special EDB predicate \id{true} without arguments in the rule body.

\begin{definition}[Answer Predicate, Goal Rule, and Query]
We assume that an IDB predicate $\id{answer}$ is distinguished as ``main'' predicate. It must not appear in the body of a program rule. A rule with the predicate $\id{answer}$ in the head is called goal rule and represents a query to the logic program.
\end{definition}

The goal of query evaluation is to determine
the $\id{answer}$-facts which are derivable
from program and database together.

\section{Deduction Method}
\label{sec:deduct}

The deduction method uses sequences of states to compute facts of the answer relation.

\begin{definition}[Rule Normalization]
Let $\mathit{VAR}$ be the set of variables of a rule $\Rule$, and let this set be ordered by the occurrence of its elements in $\Rule$: $\mathit{VAR} = \lbrace V_{i} \mid i \in \lbrace 0, \ldots , |\mathit{VAR}|-1\rbrace\rbrace$ where $V_{j} < V_{k}$ iff the first occurrence of $V_j$ is before the first occurrence of $V_k$ and $j < k$ iff $V_{j} < V_{k}$ ($j, k \in \lbrace 0, \ldots , |\mathit{VAR}|-1\rbrace$). Let further $\mathcal{X} = \lbrace \id{X}_{i} \mid i \in \bbbn \cup \lbrace 0 \rbrace \rbrace$ be an ordered set of variables: $\id{X}_{j} < \id{X}_{k}$ iff $j < k$ ($j, k \in \lbrace \bbbn \cup \lbrace 0\rbrace$). Then rule $R$ is normalized by substituting every $V_{i} \in \mathit{VAR}$ by $\id{X}_{i} \in \mathcal{X}$ ($i \in \lbrace 0, \ldots , |\mathit{VAR}|-1\rbrace$).
\end{definition}

\begin{definition}[State]
A state is a set of normalized rules.
\end{definition}

\begin{definition}[Selection Function]
A selection function chooses for every rule $\litA\lif\litB_1\until\litB_n$
with $n\geq 1$
an index~$i\in\lbrace 1\until n\rbrace$
(i.\,e.~a body literal).
\end{definition}

In every rule of a state one body literal is selected; for simplicity of presentation it is assumed that this is the leftmost body literal. However, we note that in the database context the selection function is an important optimization parameter. 
So a real implementation will use a selection function
that tries to make use of input constants
and possibly existing indexes or other database access structures.

From the rules in a state new rules are derived by two basic derivation steps that are already described in \cite{PW83}, a ``downward'' \emph{instantiation} and an ``upward'' \emph{reduction} step. In a way, this can be viewed as splitting up the SLD\--re\-so\-lu\-tion step, which avoids deriving rules with arbitrary length. A derived rule is first normalized before added to a state.

If the selected literal of a rule in the state unifies with the head of a program rule, an \emph{instance} of the program rule is derived by renaming all variables in the program rule and applying the most general unifier of the selected literal and the program rule head to the program rule. Thus, instantiation corresponds to calling an IDB predicate as in Prolog's four port box model. Several instances can be derived from the same rule. 

\begin{definition}[Instance, Instantiation] 
	Let $\Rule = \litA \leftarrow \litB_1 \andUntil \litB_{n}$ be a rule of the program and $\id{K} \leftarrow\id{L}_{1}\andUntil \id{L}_{m}, m > 0,$ a rule in the state with selected literal $\id{L}_{1}$. Let $\Rule'$ be the rule resulting from $\Rule$ by a renaming of variables so that no variable in $\Rule'$ occurs in a rule in the state, i.\,e. there is a substitution $\theta$ so that $\Rule' = \doSubst{\theta}{\Rule} = \litA' \leftarrow \litB'_{1} \andUntil \litB'_{n}$. A rule $\Rule''$ is an instance of $\Rule$ iff $\id{L}_{1}$ and $\litA'$ are unifiable with most general unifier $\sigma$ and $\Rule'' = \doSubst{\sigma}{\Rule'}$. 
\end{definition}

A reduction is performed with a fact, either of the database or a derived one. If there is a derived rule with a selected literal that unifies with the fact, this rule is \emph{reduced} by the fact and a new rule, the \emph{reduct}, is created by applying the most general unifier and removing the selected literal. When the last body literal is removed by reduction,
an IDB fact results. Thus, reduction is a special case of resolution
with a fact. Again, one fact can be used for several reductions.

\begin{definition}[Reduct, Reduction]
	Given a derived rule $\Rule = \litA \leftarrow \litB_1 \wedge \litB_{2} \andUntil \litB_{n}$ with selected literal $\litB_{1}$ and a fact $\fact$ in the database or in the state, the rule $\Rule'$ is a reduct of $\Rule$ iff $\litB_{1}$ and $\fact$ are unifiable with most general unifier $\sigma$ such that $\Rule' = \doSubst{\sigma}{(\litA \leftarrow \litB_{2} \andUntil \litB_{n})}$.
	The corresponding derivation step is called reduction, $\fact$ reduces $\Rule$ to $\Rule'$.
\end{definition}

\begin{definition}[Initial State]
The initial state consists of the goal rule and all rules that can be iteratively derived by instantiation.
\end{definition}

At a state transition, exactly one EDB fact is used to compute the successor state.

\begin{definition}[Dependency-Relation of Rules]
A rule $\Rule$ depends directly on a rule $\Rule'$
iff the selected literal in $\Rule$ is unifiable
with the head literal of $\Rule'$. 
A rule $\Rule$ depends on a rule $\Rule'$
with respect to a state $\mathcal{S}$
iff there are rules $\Rule_1, \ldots, \Rule_n \in \state$
such that $\Rule_1 = \Rule$,
each $\Rule_i, i=1, \ldots, n-1,$ depends directly on $\Rule_{i+1}$,
and $\Rule_n$ depends directly on $\Rule'$
(note that $\Rule'$ does not have to be contained in $\state$).
\end{definition}

\begin{definition}[Successor State]
\label{def:succState}
Let a program~$\prog$, a database~$\db$,
a state~$\state$,
and a fact~$\fact\in\db$ be given.
The successor state~$\state'$ is constructed as follows:
\begin{enumerate}
\item
First,
$\state'$ is initialized
with all rules that result from reduction
applied to rules in~$\state$ with fact~$\fact$.
If the result is empty,
there is no successor state.
\item
If $\state'$ now contains IDB facts,
reduction is applied repeatedly
to rules in $\state$ with facts in~$\state'$
and the results are inserted into~$\state'$
until nothing changes.
\item
Then instantiation is applied iteratively
to each rule~$\Rule\in\state'$ with a selected IDB-literal.
All instances are added to the successor state.
\item
Finally,
rules~$\Rule\in\state$ 
that depend (with respect to~$\state$) 
on a rule with at least one body literal in~$\state'$ 
are copied to~$\state'$. The copied rules are those that still have a chance of being reduced by an IDB fact.
\end{enumerate}
\end{definition}

\begin{definition}[State Sequence]
States $\state_1, \ldots, \state_n$ form a state sequence iff every $\state_{i+1}$ is the successor state for $\state_i$ and a fact $\fact_i$ of the database, $i=1, \ldots, n-1$.
\end{definition}

\begin{definition}[Computed Answers]
A fact~$\answerPred(c_1, \ldots, c_m)$ is computed
if there is a state sequence~$\state_1, \ldots, \state_n$
such that $\state_1$ is the initial state and
$\answerPred(c_1, \ldots, c_m)\in\state_n$.
\end{definition}

There can only be finitely many different states
for a given program~$\prog$ and database~$\db$ for the following reasons:
\begin{itemize}
\item The state contains only predicates
and constants occurring in the finite set~$\prog\union\db$.
\item No derived rule can become longer than the longest program rule.
\item A state does not contain two rules that differ only in the names of their variables.
\end{itemize}
However,
the state sequence could be cyclic,
so one must check whether a newly constructed state
is indeed new. 
Of course,
optimizations are possible and subject of our further research. 

\begin{example}
\label{ex:1}
Let the left recursive version
of the standard transitive closure program be given:
\begin{quote}
$\begin{array}{@{}rlcl@{}}
\lineno{1}&\path(\varA,\varB)&\lif&\edge(\varA,\varB).\\
\lineno{2}&\path(\varA,\varB)&\lif&\path(\varA,\varC)\land\edge(\varC,\varB).
\end{array}$
\end{quote}
Let the database be
\begin{quote}
$\begin{array}{@{}rl@{}}
\lineno{3}&\edge(1,2).\\
\lineno{4}&\edge(2,3).
\end{array}$
\end{quote}
Now let the following goal rule be given:
\begin{quote}
$\begin{array}{@{}rlcl@{}}
\lineno{5}&\answerPred(\varA)&\lif&\path(1,\varA).
\end{array}$
\end{quote}
The initial state~$\state_0$ consists of the goal rule
plus rules added by instantiation:
\begin{quote}
$\begin{array}{@{}rlcll@{}}
\lineno{6}&\answerPred(\varA)&\lif&\path(1,\varA).
	&\comment{goal \lineref{5}}\\
\lineno{7}&\path(1,\varA)&\lif&\edge(1,\varA).
	&\comment{inst.~of~\lineref{1} because of~\lineref{6}}\\
\lineno{8}&\path(1,\varA)&\lif&\path(1,\varB)\land\edge(\varB,\varA).
	&\comment{inst.~of~\lineref{2} because of~\lineref{6}}
\end{array}$
\end{quote}
Rule \lineref{8} also calls for instantiation
but that gives again~\lineref{7} and~\lineref{8}.

Now there is only one database fact, $\edge(1,2)$, that leads to a successor state, and by reducing with this fact we reach state~$\state_1$:
\begin{quote}
$\begin{array}{@{}rlcll@{}}
\lineno{9}&\path(1,2).&&
	&\comment{Reduction of~\lineref{7} with \lineref{3}}\\
\lineno{10}&\answerPred(2).&&
	&\comment{Reduction of~\lineref{6} with \lineref{9}}\\
\lineno{11}&\path(1,\varA)&\lif&\edge(2,\varA).
	&\comment{Reduction of~\lineref{8} with \lineref{9}}\\
\lineno{12}&\answerPred(\varA)&\lif&\path(1,\varA).
	&\comment{Copy of~\lineref{6} because of~\lineref{11}}\\
\lineno{13}&\path(1,\varA)&\lif&\path(1,\varB)\land\edge(\varB,\varA).
	&\comment{Copy of~\lineref{8} because of~\lineref{11}}\\
\end{array}\kern-10pt$
\end{quote}
Again, reduction with only one database fact, $\edge(2,3)$, is possible and
gives the state~$\state_2$:
\begin{quote}
$\begin{array}{@{}rlcll@{}}
\lineno{14}&\path(1,3).&&
	&\comment{Reduction of~\lineref{11} with \lineref{4}}\\
\lineno{15}&\answerPred(3).&&
	&\comment{Reduction of~\lineref{12} with \lineref{14}}\\
\lineno{16}&\path(1,\varA)&\lif&\edge(3,\varA).
	&\comment{Reduction of~\lineref{13} with \lineref{14}}\\
\lineno{17}&\answerPred(\varA)&\lif&\path(1,\varA).
	&\comment{Copy of~\lineref{12} because of~\lineref{16}}\\
\lineno{18}&\path(1,\varA)&\lif&\path(1,\varB)\land\edge(\varB,\varA).
	&\comment{Copy of~\lineref{13} because of~\lineref{16}}\\
\end{array}\kern-15pt$
\end{quote}
No more reductions with database facts can be applied to rules in $\state_2$.
\qed
\end{example}

\begin{theorem}[Correctness]
Let a program $\prog$ and a database $\db$ be given. Every computed answer is indeed a logical consequence
of $\prog\union\db$.
\end{theorem}

\begin{proof}
This is easy:
Each step (reduction and instantiation)
is a logical consequence
of $\prog\union\db$ and the previously computed rules.
\qed
\end{proof}

\begin{theorem}[Completeness]
For every ground substitution $\theta$ such that \linebreak$\doSubst{\theta}{\answerPred(\id{X}_1,\ldots, \id{X}_q)}$ is a logical consequence of the program and the database, $\doSubst{\theta}{\answerPred(\id{X}_1,\ldots, \id{X}_q)}$ is computed.
\end{theorem}

The completeness theorem is a corollary of the following lemma, if $\state_0$ is the initial state and the rule considered is the goal rule.

\begin{lemma}
Let a program $\prog$ and a database $\db$ be given. If a state~$\state_0$ contains a rule $\Rule = \litA\lif\litB_1\andUntil\litB_n$
and there is a ground substitution~$\subst$
such that each $\doSubst{\subst}{\litB_i}$ is a logical consequence
of $\prog\union\db$,
then there is a state sequence $\state_0, \state_1, \ldots, \state_m$ such that $\doSubst{\subst}{\litA}$
is contained in $\state_m$.
Furthermore,
any rule~$\Rule'\in\state_0$
that depends on $\Rule$
is contained in every state $\state_1, \ldots, \state_{m-1}$.
\end{lemma}

\begin{proof}
Since the~$\doSubst{\subst}{\litB_i}$ are logical consequences
of~$\prog\union\db$,
they are contained in the least fixpoint of~$\id T_{\prog\union\db}$,
and because there are no function symbols,
this is reached after a finite number of iterations.
The proof is by induction on the maximum (over~$i$, $i=1,\ldots,n$)
of the number of steps needed to derive~$\doSubst{\subst}{\litB_i}$ with the $\id T_{\prog\cup\db}$-operator.

If this is~1,
all $\doSubst{\subst}{\litB_i}$ are contained in~$\db$. For proving the first step of induction, there is a second induction on $n$ (the number of body literals in rule $\Rule$). If this is 1, then $\Rule = \litA\lif\litB_1$, $\state_m = \state_1$ is the direct successor of $\state_0$ and contains $\doSubst{\theta}{\litA}$. $\state_0$ contains all rules depending on $R$. 
Now assume that the theorem is proven for rules $\Rule$ with $n$ body literals $\litB_i$, all of them with EDB predicate. Assume further that $\state_0$ contains $\litA\lif\litB_1\wedge\litB_2\andUntil\litB_{n+1}$. If $\litB_1$ is selected, there exists a direct successor state of $\state_0$ with the fact $\doSubst{\theta_1}{\litB_1}$ that contains $\doSubst{\theta_1}{(\litA\lif\litB_2\andUntil\litB_{n+1})}$ and all rules in $\state_0$ depending on this rule, 
where $\theta_1$ is $\theta$ restricted to the variables occurring in $\litB_1$. From the hypothesis of the second induction the theorem follows, and $m = n+1$.

Now assume that the theorem is proven for all cases where the body literals $\doSubst{\theta}{\litB_i}$ are derivable after at most $k$ steps of the $\id T_{\prog\cup\db}$-operator. This means that $\doSubst{\theta}{\litA}$ can be derived after $k+1$ steps and computed with a state sequence $\state_0, \ldots, \state_m$ of length $m+1$. 

For the induction step, suppose that all $\doSubst{\theta}{\litB_i}$ are derivable after at most $k+1$ steps of the $\id T_{\prog\cup\db}$-operator. Again, to prove the theorem there is an induction on the number of $\litB_i$ in $\Rule$. If this is 1, then $\Rule = \litA\lif\litB_1$. We consider only the case that $\litB_1$ is an IDB literal (the other case is already shown above). In this case, an instantiation is performed, so 
$\state_0$ contains a rule $\litB\lif\litC_1\andUntil\litC_l$ where $\litB$ unifies with $\litB_1$. This rule is either the new instance or a rule already present in the state and equal to the new instance. From the inductive hypothesis follows that $\doSubst{\theta_1}{\litB}$ (where $\theta_1$ is $\theta$ restricted to the variables occurring in $\litB$) can be computed with a state sequence $\state_0, \ldots, \state_m$
, and that $\state_{m-1}$ contains the rule $\Rule = \litA\lif\litB_1$ which depends on $\litB$. Since $\litB_1$ is the only body literal, $\theta_1 = \theta$. Thus, in state $\state_m$ a reduction with $\doSubst{\theta}{\litB}$ and $\litA\lif\litB_1\in\state_{m-1}$ can be performed so that $\state_m$ also contains $\doSubst{\theta}{\litA}.$

If $\state_0$ contains a rule $\litA\lif\litB_1\wedge\litB_2\andUntil\litB_n$ where $\litB_1$ is selected and an IDB literal, again it also contains a rule $\litB\lif\litC_1\andUntil\litC_l$. A state sequence $\state_0, \ldots, \state_{m'}$ can be computed where $\state_{m'}$ contains $\doSubst{\theta_1}{\litB}$, $\state_{m'-1}$ contains $\litA\lif\litB_1\wedge\litB_2\andUntil\litB_n$ and, after a reduction with the fact $\doSubst{\theta_1}{\litB}$, $\state_{m'}$ contains $\doSubst{\theta_1}{\litA\lif\litB_2\andUntil\litB_n}$ as well as all rules in $\state_{m'-1}$ depending on this rule. Finally, from both inductive hypotheses the theorem follows.
\qed
\end{proof}

\section{Partial Evaluation}
\label{sec:partEval}
Especially in database context, the facts of the extensional database might not be known before execution time, and as the aim is to compile a program beforehand, an abstraction from actual data values must be developed. For this purpose, abstract values taken from an infinite set of symbolic constant values, $\mathcal{V}$, are used instead of the data values that are known only at execution time. Constants in program rules and in the query are not substituted, so no symbolic value may occur in a program rule or in the query. Via a partial evaluation an automaton can be constructed that models the process of query evaluation.

We only need to redefine a state transition. For a given state, create a set with all selected EDB literals in the state that are not equal to each other. Two literals are considered equal if they differ only in the names of their variables. For every literal in this set there is a state transition assigned to it which is labeled with the literal. Choose a literal and substitute its variables with new symbolic values that have not yet been used elsewhere. The resulting ``symbolic fact'' represents all facts that could be obtained from a query to the corresponding EDB relation at execution time. Thus, a state transition can be viewed as a data retrieving interface. Now reduce all possible rules in the  
state with the symbolic fact and add the derived rules to 
a new state. 
If 
a state contains a fact of the answer relation, it is a \emph{final state}. 
In the same way as an EDB literal with symbolic values can be viewed as a representative of a set of facts, a state with symbolic values can be viewed as a representative of a set of states that depends on the actual data values.

\begin{example}
\label{ex:symb}
Consider again the transitive closure program with the same goal rule.
The database is irrelevant now, only name and arity of EDB predicates are needed. 
Let the set $\mathcal{V}$ of symbolic values be $\lbrace \id{c}_0, \id{c}_1, \ldots \rbrace$.
For the initial state~$\state_0$ there are no differences. 
\begin{quote}
$\begin{array}{@{}rlcll@{}}
\lineno{6}&\answerPred(\varA)&\lif&\path(1,\varA).
	&\\
\lineno{7}&\path(1,\varA)&\lif&\edge(1,\varA).
	&\\
\lineno{8}&\path(1,\varA)&\lif&\path(1,\varB)\land\edge(\varB,\varA).
	&
\end{array}$
\end{quote}
In rule [7], $\edge(1,\varA)$ is selected, so $\varA$ is substituted by the symbolic value $\id{c}_0$. A state transition labeled with $\edge(1,\varA)$ is created, and a transition with the symbolic fact $\edge(1,\id{c}_0)$ gives $\state_1$:
\begin{quote}
	$\begin{array}{@{}rlcll@{}}
		\lineno{9}&\path(1,\id{c}_0).&&
		&\comment{Reduction of~\lineref{7} with fact}\\
		\lineno{10}&\answerPred(\id{c}_0).&&
		&\comment{Reduction of~\lineref{6} with \lineref{9}}\\
		\lineno{11}&\path(1,\varA)&\lif&\edge(\id{c}_0,\varA).
			&\comment{Reduction of~\lineref{8} with \lineref{9}}\\
		\lineno{12}&\answerPred(\varA)&\lif&\path(1,\varA).
			&\comment{Copy of~\lineref{6} because of~\lineref{11}}\\
		\lineno{13}&\path(1,\varA)&\lif&\path(1,\varB)\land\edge(\varB,\varA).
			&\comment{Copy of~\lineref{8} because of~\lineref{11}}\\
	\end{array}\kern-10pt$
\end{quote}
From $\state_1$ the successor state $\state_2$ is reached by a transition with the symbolic fact $\edge(\id{c}_0,\id{c}_1)$, labeled with the literal $\edge(\id{c}_0,\varA)$.
\begin{quote}
	$\begin{array}{@{}rlcll@{}}
		\lineno{14}&\path(1,\id{c}_1).&&
		&\comment{Reduction of~\lineref{11} with fact}\\
		\lineno{15}&\answerPred(\id{c}_1).&&
		&\comment{Reduction of~\lineref{12} with \lineref{14}}\\
		\lineno{16}&\path(1,\id{X}_0)&\lif&\edge(\id{c}_1,\id{X}_0).
			&\comment{Reduction of~\lineref{13} with \lineref{14}}\\
		\lineno{17}&\answerPred(\id{X}_0)&\lif&\path(1,\id{X}_0).
			&\comment{Copy of~\lineref{12} because of~\lineref{16}}\\
		\lineno{18}&\path(1,\id{X}_0)&\lif&\path(1,\id{X}_1)\land\edge(\id{X}_1,\id{X}_0).
			&\comment{Copy of~\lineref{13} because of~\lineref{16}}\\
	\end{array}\kern-10pt$
\end{quote}
All following states resemble $\state_1$ but have symbolic constants $\id{c}_2, \id{c}_3, \ldots$ instead of $\id{c}_1$.
\qed
\end{example}

For programs without recursive rules, the above state construction algorithm works well. In the other cases, as in our example, there will be infinitely many states since there are infinitely many symbolic values to be used at state transitions. Nevertheless, a part of these cases can be handled by trying to find finite many \emph{equivalence classes} of states. It may be noticed that states are generated which have a similar structure but different symbolic values. Similar states generate again similar states because the same derivation steps are applied to similar sets of rules. Therefore they can be combined in one equivalence class of states.

\begin{definition}[Equivalent States]
Let two states $\state_1$ and $\state_2$ be given. Let further $\mathit{SYMB}_{1}$ be the set of symbolic values occurring in $\state_1$ and $\mathit{SYMB}_{2}$ the set of symbolic values occurring in $\state_2$.
$\state_1$ is equivalent to $\state_2$ iff a bijective mapping $\id{map}$ from $\mathit{SYMB}_{1}$ to $\mathit{SYMB}_{2}$ exists so that the state $\state_1'$ obtained by replacing every symbolic value $\id{v}$ in $\state_1$ by $\id{map}(\id{v})$ is equal to $\state_2$.
\end{definition}

The construction of an automaton with partial evaluation is straightforward. The states of this automaton represent equivalence classes of those states that are constructed during the derivation process. Consequently, when a state is constructed for which an equivalent state already exists, these states are fused to one state in the automaton. 

\begin{example}
\label{ex:equiv}
Consider again the transitive closure program with the goal rule and with symbolic values (Examples \ref{ex:1} and \ref{ex:symb}). 
It is visible that $\state_1$ and $\state_2$ are equivalent: if in $\state_2$ the symbolic value $\id{c}_1$ is mapped to $\id{c}_0$, both rule sets are equal. With this, the following state transition function results:
\begin{quote}
	\begin{tabular}{llll}
			$\delta$($\state_{0}\!$&, $\edge(1, \id{X}_0)$)&$=$&$\state_{1}$\\
			$\delta$($\state_{1}\!$&, $\edge(\id{c}_0, \id{X}_0)$)&$=$&$\state_{1}$\\
	\end{tabular}\\
\end{quote}
\qed
\end{example}

In certain cases it happens that arbitrarily many different rules with the same structure of constants and variables but different values are accumulated in 
one state. 
In these cases the process of partial evaluation and automata construction does not terminate. A part of these cases results from tail recursive program rules. The problem with a tail recursive rule is that, starting from the last literal of the rule, arbitrarily long instantiation chains are created which have to be kept in the state for reduction. 
These cases can be dealt with by introducing an additional derivation step and performing a \emph{resolution} step instead of an instantiation when processing the last literal of a rule. 

\begin{definition}[Extension of Deduction Method]
The algorithm described in Definition \ref{def:succState} is extended as follows:
\begin{enumerate}
\item
A reduction with a fact $\fact$ the predicate of which is an EDB predicate is not affected and performed as described above.
\item
Reductions with IDB facts are applied repeatedly, but only to rules $\Rule = \litA \leftarrow \litB_1 \wedge \litB_{2} \andUntil \litB_{n}$ where $n > 1$. 
\item
Instantiations are only applied to rules $\Rule = \litA \lif \litB_1 \andUntil \litB_n$ with selected literal $\litB_1$ where the predicate of $\litB_1$ is an IDB predicate and $n > 1$. Otherwise, if $n = 1$, new rules are derived by applying a resolution step to $\litA \lif \litB_1$ and program rules where the head literal unifies with $\litB_1$. The derivation step is therefore called \emph{last literal resolution}.
\item
If a rule $\Rule = \litA \lif \litB_1 \andUntil \litB_n$ depends on a rule $\Rule'$ in the successor state it is only copied to the successor state if $n > 1$.
\end{enumerate}
The initial state consists of the goal rule and all rules that can be iteratively derived by instantiation and last literal resolution. 
\end{definition}

\begin{example}
\label{ex:res}
Consider the tail recursive version of the transitive closure program:
\begin{quote}
$\begin{array}{@{}rlcl@{}}
\lineno{1}&\path(\varA,\varB)&\lif&\edge(\varA,\varB).\\
\lineno{2}&\path(\varA,\varB)&\lif&\edge(\varA,\varC)\land\path(\varC,\varB).
\end{array}$
\end{quote}
With the goal rule
\begin{quote}
$\begin{array}{@{}rlcl@{}}
\lineno{3}&\answerPred(\varA)&\lif&\path(1,\varA).
\end{array}$
\end{quote}
the initial state is:
\begin{quote}
	$\begin{array}{@{}rlcll@{}}
		\lineno{4}&\answerPred(\varA)&\lif&\path(1,\varA).
		&\comment{goal rule}\\
		\lineno{5}&\answerPred(\varA)&\lif&\edge(1,\varA).
		&\comment{last literal resolution of \lineref{4}}\\
		\lineno{6}&\answerPred(\varA)&\lif&\edge(1,\varB)\land\path(\varB,\varA).
		&\comment{last literal resolution of \lineref{4}}\\
	\end{array}\kern-10pt$
\end{quote}
A transition with the symbolic fact $\edge(1,\id{c}_0)$ gives $\state_1$:
\begin{quote}
	$\begin{array}{@{}rlcll@{}}
		\lineno{7}&\answerPred(\id{c}_0).&&
		&\comment{reduction of \lineref{5} with fact}\\
		\lineno{8}&\answerPred(\varA)&\lif&\path(\id{c}_0,\varA).
		&\comment{reduction of \lineref{6} with fact}\\
		\lineno{9}&\answerPred(\varA)&\lif&\edge(\id{c}_0,\varA).
		&\comment{last literal resolution of \lineref{8}}\\
		\lineno{10}&\answerPred(\varA)&\lif&\edge(\id{c}_0,\varB)\land\path(\varB,\varA).
		&\comment{last literal resolution of \lineref{8}}\\
	\end{array}\kern-10pt$
\end{quote}
The following states are equivalent to this state. The transition function is the same as for the left recursive program:
\begin{quote}
	\begin{tabular}{llll}
			$\delta$($\state_{0}\!$&, $\edge(1, \varA)$)&$=$&$\state_{1}$\\
			$\delta$($\state_{1}\!$&, $\edge(\id{c}_0, \varA)$)&$=$&$\state_{1}$\\
	\end{tabular}
\end{quote}
\qed
\end{example}

In order to achieve termination, the other cases of infinitely growing states have to be excluded for the time being. For this the notion of the \emph{schema} of a rule is introduced, and the set of \emph{valid} states is restricted to those states that do not contain two rules with the same schema. {\sc Porter} uses a similar definition for a schema in \cite{porter2}.

\begin{definition}[Schema of a Rule]
Let a normalized rule $\Rule$ be given, and let $\mathit{CONST}_{\mathcal{V}}$ be the set of those constants in $R$ that do not occur in the program, which means they are symbolic values of the set $\mathcal{V}$. Let $\mathit{CONST}_{\mathcal{V}}$ be ordered by the occurrence of its elements in $\Rule$: $\mathit{CONST}_{\mathcal{V}} = \lbrace c_{i} \mid i = 0, \ldots , |\mathit{CONST}_{\mathcal{V}}|-1 \rbrace$ where $c_{j} < c_{k}$ iff the first occurrence of $c_j$ is before the first occurrence of $c_k$ and $j < k$ iff $c_{j} < c_{k}$ ($j, k \in \lbrace 0, \ldots , |\mathit{CONST}_{\mathcal{V}}|-1\rbrace$). Let further $\mathcal{B} = \lbrace \id{b}_{i} \mid i \in \bbbn \cup \lbrace 0 \rbrace \rbrace$ be an ordered set of constants: $\id{b}_{j} < \id{b}_{k}$ iff $j < k$ ($j, k \in \lbrace \bbbn \cup \lbrace 0\rbrace$). The schema of $\Rule$ is the rule obtained by replacing every $c_{i} \in \mathit{CONST}_{\mathcal{V}}$ by $\id{b}_{i} \in \mathcal{B}$ ($i \in \lbrace 0, \ldots , |\mathit{CONST}_{\mathcal{V}}|-1\rbrace$).
\end{definition}

\begin{example}
Consider state $\state_1$ of the tail recursive program in Example \ref{ex:res}. The schemata of the rules in this state are as follows: 
\begin{quote}
	$\begin{array}{@{}rlcll@{}}
		\lineno{7}&\answerPred(\id{b}_0).&&
		&\\
		\lineno{8}&\answerPred(\id X_0)&\lif&\path(\id{b}_0,\id X_0).
		&\\
		\lineno{9}&\answerPred(\id X_0)&\lif&\edge(\id{b}_0,\id X_0).
		&\\
		\lineno{10}&\answerPred(\id X_0)&\lif&\edge(\id{b}_0,\id X_1)\land\path(\id X_1,\id X_0).
		&\\
	\end{array}\kern-10pt$
\end{quote}
\qed
\end{example}

\begin{definition}[Valid State]
Let $\state$ be a state 
and $\mathcal{SCH}$ be the set of schemata of the rules in $\state$. The state $\state$ is valid iff $|\mathcal{SCH}| = |\state|$.
\end{definition}

For a given program that meets all requirements mentioned above there are only finitely many possibilities to create valid states, so partial evaluation is guaranteed to terminate.

An implementation of the automaton will use states where the symbolic values are replaced by assignable variables that hold the actual data values. Explicit constants can be included in the target code and no derivations have to be performed so that the runtime states are very compact and the main task at state transitions should be accessing and selecting the data. Different results for the \id{answer} predicate can be obtained by backtracking or by concurrent processing of alternative transitions. 

\section{Conclusion}
We have presented an algorithm for efficient query evaluation and preprocessing of func\-tion\--free logic programs based on Earley Deduction. The algorithm can process non\--re\-cur\-sive, left- and tail\--re\-cur\-sive rules and has been proven to be correct, complete, and terminating. While Earley Deduction can in principle be used for arbitrary logic programs, still the basic algorithm presented here is already an improvement to it because it processes several derivations for one fact in one step. There is a special optimization potential
when it can be proven that only one fact is applicable in a state
and we do not have to check whether there is a cycle in the state sequence. Subjects of our future research include further optimizations for special applications and an efficient implementation of the generated automaton. 

Further material, including a demo program (written in SWI Prolog) showing the state sequences for a given program and query, is available at 
\\
\texttt{http://dbs.informatik.uni-halle.de/Earley}.

\bibliographystyle{splncs03}
\bibliography{wlp12.bib}

\end{document}